\documentclass[reqno]{amsart}
\usepackage{graphicx}
\usepackage{enumitem}
\usepackage{todonotes}
\usepackage[numbers,sort&compress]{natbib}
\usepackage{hyperref}
\usepackage{soul}
\usepackage[foot]{amsaddr}
\usepackage[final,color,notref,notcite]{showkeys}
\usepackage{caption}

\usepackage[charter]{mathdesign}
\usepackage[scr=rsfso, scrscaled=.98]{mathalfa}

\definecolor{labelkey}{rgb}{0,.56,.7}
\presetkeys{todonotes}{color=red}{}

\DeclareMathAlphabet{\pazocal}{OMS}{zplm}{m}{n}   

\newcommand{\Ocal}{\pazocal{O}}

\newcommand{\Lcal}{\pazocal{L}}

\newtheorem{thm}{Theorem}
\newtheorem{prop}[thm]{Proposition}
\newtheorem{lem}[thm]{Lemma}

\theoremstyle{remark}
\newtheorem*{rem}{\bf Remark}
\newtheorem{defi}{\bf Definition}

\newtheorem{exaa}{\bf Example}

\newenvironment{exa}
  {\pushQED{\qed}\exaa}
  {\popQED\endexaa}

\newcommand*{\at}{@}

\newcommand{\nn}{\nonumber}

\def\dg{\dagger}
\def\df{\overset{\mathrm{df}}{=}}
\def\ba{\boldsymbol{\a}}
\newcommand{\ket}[1]{\mathop{|#1\rangle}\nolimits}
\newcommand{\bra}[1]{\mathop{\left<#1\,\right|}\nolimits}

\newcommand{\kbr}[2]{| #1\rangle\!\langle #2 |}

\def\ra{\rangle}
\def\l{\langle}

\newcommand{\diff}[2]{\mathrm{\,d}^#1 #2}             
\newcommand{\dif}[1]{\mathrm{\,d} #1}             
\newcommand{\difff}[1]{\mathrm{\,d}^3 #1}             

\def\a{\alpha}

\def\g{\gamma}
\def\G{\Gamma}
\def\d{\delta}

\def\om{\omega}
\def\Om{\Omega}
\def\s{\sigma}

\def\la{\lambda}

\def\bbZ{\mathbb{Z}}

\def\bbE{\mathbb{E}}

\sodef\so{}{.065em}{.4em plus1em}{2em plus.1em minus.1em}

\begin{document}

\title[Unitary evolution of a pair of Unruh-DeWitt detectors calculated efficiently...]{{Unitary evolution of a pair of Unruh-DeWitt detectors calculated efficiently to an arbitrary perturbative order}}

\begin{abstract}
Unruh-DeWitt Hamiltonian couples a scalar field with a two-level atom serving as a particle detector model. Two such detectors held by different observers following general trajectories can be used to study entanglement behavior in quantum field theory. Lacking other methods, the unitary evolution  must be studied perturbatively which is considerably time-consuming even to a low perturbative order. Here we completely solve the problem and present a simple algorithm for a perturbative calculation based on a solution of a system of linear Diophantine equations. The algorithm runs polynomially with the perturbative order. This should be contrasted with the number of perturbative contributions of the scalar $\phi^4$ theory that is known to grow factorially.

Speaking of the  $\phi^4$ model, a welcome collateral result is obtained to mechanically (almost mindlessly) calculate the interacting scalar $\phi^n$  theory without resorting to Feynman diagrams. We demonstrate it on a typical textbook example of two interacting fields for $n=3,4$.
\end{abstract}


\author{Kamil Br\'adler}

\email{kbradler\at uottawa.ca}

\address{Department of Mathematics and Statistics, University of Ottawa, Ottawa, Canada}

\maketitle

\thispagestyle{empty}
\allowdisplaybreaks

\section{Introduction}\label{sec:intro}

An interaction model  between a real scalar field and a two-level atom called Unruh-DeWitt (UDW) particle detector was conceived~\cite{unruh1976notes}, improved~\cite{dewitt1979quantum} and studied in many physical situations~\cite{svaiter1992inertial,higuchi1993uniformly,ver2009entangling,reznik2005violating,schlicht2004considerations,lin2006accelerated,louko2006often,sriramkumar1996finite,barbado2012unruh,hummer2016renormalized,franson2008generation,bradler2016absolutely}. A UDW detector consists of a two-level atom linearly coupled to a scalar field. Research in this area typically focuses on a single UDW detector to probe the field or two UDW detectors to study bipartite correlations with a special interest in the vacuum entanglement. The calculation is of a perturbative character but to the author's knowledge, the Unruh-DeWitt Hamiltonian has not been (dis)proved to be exactly solvable~\cite{birrell1984quantum}. The main result of this paper is the next best thing. We also approach the problem perturbatively and provide an expansion to an arbitrary order of a coupling constant and so we can study the field-detector(s) interaction in the strong coupling regime. Crucially, unlike a typical situation in quantum field theory, we show that the number of perturbative terms grows polynomially with the perturbative order. This is in a sharp contrast with the textbook problem of the interacting scalar $\phi^n$ theory (and other more realistic models) where the number of terms grows factorially~\cite{hurst1952enumeration,bender1976statistical,cvitanovic1978number,kleinert2000recursive}. Additionally, the perturbative series are typically asymptotic and this limits even the perturbative approach itself. As one reaches  the perturbative order of approximately the inverse of the coupling strength the series summands start to increase. The series diverges~\cite{dyson1952divergence,Lipatov1976ny} and  resummation techniques must be used~\cite{zinn1981perturbation}. In the case of two  UDW detectors  we show that there is no factorial explosion and the polynomial growth indicates no problems with  convergence. Could it be then that the actual series  for two UDW detectors we obtained here is summable or is the issue of divergence more subtle and persistent? This is an interesting open question probably waiting to be rigorously settled but the results presented in this paper will prove to be useful irrespective of the answer.

Surprisingly, a part of our UDW calculation turns out to be useful for a perturbative calculation of the scattering amplitudes in the interacting scalar $\phi^n$ theory. A typical approach for the high-order loop corrections is to enumerate all participating Feynman diagrams~\cite{kajantie2002simple,kleinert2000recursive}.  Here we circumvent this step and instead use the statistics cousin of Wick's theorem~\cite{wick1950evaluation} known as Isserlis' theorem~\cite{isserlis1918formula} to directly calculate the perturbative contributions. This is possible due to an easy way of finding all (non-negative) solutions of a system of linear Diophantine equations Isserlis' theorem leads to. As a result, the scattering amplitudes for the interacting $\phi^n$ scalar model  can be found without  the need to infer Feynman's rules~\cite{schwartz2014quantum}. Feynman's diagrams can, however, be easily recovered and we demonstrate the whole process on the textbook examples of the second order expansion of the $\phi^3$ and~$\phi^4$ models (for two interacting fields). This obviously does not remove or circumvent the fast growth of the perturbative contributions in these models but it seems that the approach coming from this direction is novel and makes the perturbative calculations (even to high orders) extremely streamlined. This may be surprising due to the fact that Diophantine equations have a well-earned reputation of being hard to solve.

In Section~\ref{sec:UDW} we recall the definition of the Unruh-DeWitt detector, take two copies of it and introduce the problem we would like to solve. The solution is presented in Section~\ref{sec:phases} and is divided in three phases. In Section~\ref{sec:phin} we show how to apply our approach to perturbatively calculate the scattering amplitudes for the $\phi^n$ scalar model.

\section{Unruh-DeWitt detector}\label{sec:UDW}

A realistic UDW detector is described by the following term
\begin{equation}\label{eq:UDWHint}
  H(\tau)=\la w(\tau)(\s^+e^{i\tau\d}+\s^-e^{-i\tau\d})\int\difff{x} f(x)\phi(t,x),
\end{equation}
and
\begin{equation}\label{eq:RealField}
  \phi(t,x)={1\over(2\pi)^3}\int{\difff{k}\over2\om_k}\big(a_ke^{-i(\om_kt-k\cdot x)}+a^\dg_ke^{i(\om_kt-k\cdot x)}\big)
\end{equation}
is a real scalar field where $(\om_k,k)$ is a four-momentum, $w(\tau)$ is a detector's window function ($\tau$ is its proper time and $t(\tau),x(\tau)$ are Minkowski coordinates), $f(x)$ a smearing function, $\s^\pm$ are the detector ladder operators, $\d$ stands for the energy gap and $\la$ for a coupling constant. The evolution operator for observer $A$ reads
\begin{align}\label{eq:2ndOrderExp}
  U_A(\tau_1,\tau_0) & = \mathsf{T}{\Big\{\exp{\big[-i\int_{\tau_0}^{\tau_1}\dif{\tau'}H_A(\tau')\big]}\Big\}},
\end{align}
where $\mathsf{T}$ is the time-ordering operator, and similarly for observer $B$.

Out task is to efficiently calculate the following unitary
\begin{equation}\label{eq:twoUDWdetectors}
U_A(\tau_1,\tau_0)\otimes U_B(\tau_2,\tau_0) = e^{-i(A_+\sigma_A^++A_-\sigma_A^-+B_+\sigma_B^++B_-\sigma_B^-)} = \sum_{n=0}^\infty \frac{(-i)^n\la^n}{n!} (A_+\sigma_A^++A_-\sigma_A^-+B_+\sigma_B^++B_-\sigma_B^-)^n
\end{equation}
for any $n$, no matter how large, and for any input atomic state. We denoted
\begin{align}
    A_+ & = w(\tau_1)e^{i\tau_1\d}\phi(t_1(\tau_1)), \\
    B_+ & = w(\tau_2)e^{i\tau_2\d}\phi(t_2(\tau_2))
\end{align}
and so $A_-=A_+^\dg,B_-=B_+^\dg$ holds. Note that we set $f(x)$ to be a delta function for convenience. The result of this paper is oblivious to the details of smearing, window functions or trajectories. The important expression is Eq.~(\ref{eq:twoUDWdetectors}) and the fact that $A_\pm,B_\pm$ are bosonic in nature. Our final goal is to express $2n$-point correlators in terms of a polynomial number of 2-point correlators and not their calculation where these details are relevant.

The matrix elements of~(\ref{eq:twoUDWdetectors}) (in the canonical basis) will be shorthanded as
\begin{equation}\label{eq:unitaryEntries}
\bra{kl} U_AU_B \ket{ij},
\end{equation}
where $i,j,k,l\in\{0,1\}$ and a tensor product implied. So for the detector density matrix we get
\begin{equation}\label{eq:omega}
\om_{AB}(\kbr{kl}{mn}) = \bra{0_M}\mathsf{T}{\big\{ \bra{kl} U_AU_B \ket{ij} \big(\bra{mn} U_AU_B \ket{ij}\big)^\dg\big\}} \ket{0_M},
\end{equation}
where $\ket{0_M}$ denotes Minkowski vacuum and $\ket{ij}$ is the detectors' initial state.

\section{Two Unruh-DeWitt detectors to an arbitrary perturbative order}\label{sec:phases}

This section contains our main result. We will derive $\om_{AB}$ for $i,j=0$ and show a trivial modification enabling us to calculate Eq.~(\ref{eq:unitaryEntries}) for any canonical basis state $\ket{ij}$. This, on the other hand, will lead to  a `standalone' unitary matrix $U_A\otimes U_B$ that can be used to find $\om_{AB}$ for \emph{any}  detector input state and thus completely solving the problem. In the first step (Phase~I) we tame the exponential number (in $n$) of summands of the core expression~Eq.~(\ref{eq:unitaryEntries}) (or~(\ref{eq:twoUDWdetectors})).
By plugging the result into~Eq.~\ref{eq:omega} we get again a polynomial amount of $2n$-point Green's functions. We will be able to split them into a product of two-point correlation functions by constructively solving a system of linear Diophantine equations with a polynomial number of non-negative solutions in $n$ (Phase~II) and calculate their multiplicities (Phase~III).

\subsection*{Phase~I}
Let us recall the elements of Pauli operators algebra representing the two-level detector(s). They satisfy $(\s^+)^2=(\s^-)^2=0$ and also $\s_A^\pm\s_B^\pm=\s_B^\pm\s_A^\pm$ together with $\s_A^\pm\s_B^\mp=\s_B^\mp\s_A^\pm$. That is, the atomic operators for the two detectors commute in the Unruh-DeWitt model.
\begin{prop}\label{prop:phaseI}
  For $n=2m$ we obtain from Eq.~(\ref{eq:twoUDWdetectors})
  \begin{align}\label{eq:expEven}
    (A_+\sigma_A^++A_-\sigma_A^-+B_+\sigma_B^++B_-\sigma_B^-)^{2m}\ket{00}
    &=\sum_{\ell=0}^m\binom{2m}{2\ell}\,A_+^\ell A_-^\ell B_+^{m-\ell}B_-^{m-\ell}\ket{00}\nn\\
    &+\sum_{\ell=1}^m\binom{2m}{2\ell-1}\,A_+^{\ell}A_-^{\ell-1} B_+^{m-\ell+1}B_-^{m-\ell}\ket{11}
  \end{align}
  and for $n=2m+1$ we get
  \begin{align}\label{eq:expOdd}
    (A_+\sigma_A^++A_-\sigma_A^-+B_+\sigma_B^++B_-\sigma_B^-)^{2m+1}\ket{00}
    &=\sum_{\ell=0}^m\binom{2m+1}{2\ell}\,A_+^\ell A_-^\ell B_+^{m-\ell+1}B_-^{m-\ell}\ket{01}\nn\\
    &+\sum_{\ell=0}^m\binom{2m+1}{2\ell}\,A_+^{m-\ell+1}A_-^{m-\ell} B_+^{\ell}B_-^{\ell}\ket{10}.
  \end{align}
\end{prop}
\begin{proof}
  The initial state $\ket{0}$ is a ground state: $\s^-\ket{0}=0$. There are two possibilities for the output states: $\ket{0}$ and $\ket{1}$ but taking into account the nilpotency property of $\s^\pm$ we can also enumerate all possible ways the output states can be obtained. It is simply
  \begin{align}\label{eq:allPaths}
    \ket{0} & = (\s^-\s^+)^{p}\ket{0}, \\
    \ket{1} & = \s^+(\s^-\s^+)^{p-1}\ket{0}
  \end{align}
  for all $p>0$. All other operator sequences result in zero. In our case we have two sets of ladder operators -- for the $A$ and $B$ Hilbert spaces and so
\begin{subequations}\label{eq:allPaths2Atoms}
  \begin{align}
    \ket{00} & = \sum_{\genfrac{}{}{0pt}{2}{\mathrm{nonzero\ products\ s.t.}}{2p+2q=n}}(\s_A^-\s_A^+)^{p}(\s_B^-\s_B^+)^{q}\ket{00}, \label{eq:allPaths2Atoms00} \\
    \ket{11} & = \sum_{\genfrac{}{}{0pt}{2}{\mathrm{nonzero\ products\ s.t.}}{2p+2q-2=n}}(\s_A^-)^{p-1}(\s_A^+)^{p}(\s_B^-)^{q-1}(\s_B^+)^{q}\ket{00}, \label{eq:allPaths2Atoms11} \\
    \ket{01} & = \sum_{\genfrac{}{}{0pt}{2}{\mathrm{nonzero\ products\ s.t.}}{2p+2q=n}}(\s_A^-)^{p}(\s_A^+)^{p}(\s_B^-)^{q-1}(\s_B^+)^{q}\ket{00}, \label{eq:allPaths2Atoms01} \\
    \ket{10} & = \sum_{\genfrac{}{}{0pt}{2}{\mathrm{nonzero\ products\ s.t.}}{2p+2q-1=n}}(\s_A^-)^{p-1}(\s_A^+)^{p}(\s_B^-)^{q}(\s_B^+)^{q}\ket{00}.\label{eq:allPaths2Atoms10}
  \end{align}
\end{subequations}
  In the first two lines $n$ is even and in the last two lines $n$ is odd. The sums should be understood as counting all possible strings of the sigma operators leading to a given output state on the left and we will now find the explicit expressions. Let's call them \emph{nonzero strings}. Given the desired output state,  the main observation is that there are always only two sigma operators that do not destroy the previous nonzero string. This is because every nonzero string acting on $\ket{00}$ can land in only one of the four possible states. For instance, if the state is $\ket{01}$ the next sigma operator can be $\s^+_A$ or $\s_B^-$ to get another nonzero string. Similarly for the other three states and so we can represent all nonzero strings as a complete binary tree of the depth $n$ where we adopt the following convention: The root node is the initial state $\ket{00}$ where the left offsprings correspond to the action of $\s_A^+$ and the right offsprings correspond to $\s_B^+$. The nodes are labeled by the resulting state $\ket{ij}$, see Fig.~\ref{fig:littlebintree}. In all the branchings that follow the left offsprings correspond to the action of $\s_A^\pm$ and the right offsprings correspond to $\s_B^\pm$.  Given the offspring node, the choice of either $+$ or $-$ in the branching is unambiguous.
    \begin{figure}[h]
    \fcolorbox{white}{white}{\includegraphics[scale=.82]{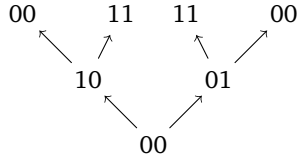}}
    \caption{All paths for the first two branchings starting from $\ket{00}_{AB}$ are depicted. The arrows indicate the action of the ladder operators $\s_{A,B}^\pm$. The resulting state is a new node.}\label{fig:littlebintree}
  \end{figure}
  \begin{figure}[t]
   \fcolorbox{white}{white}{\includegraphics[scale=.75]{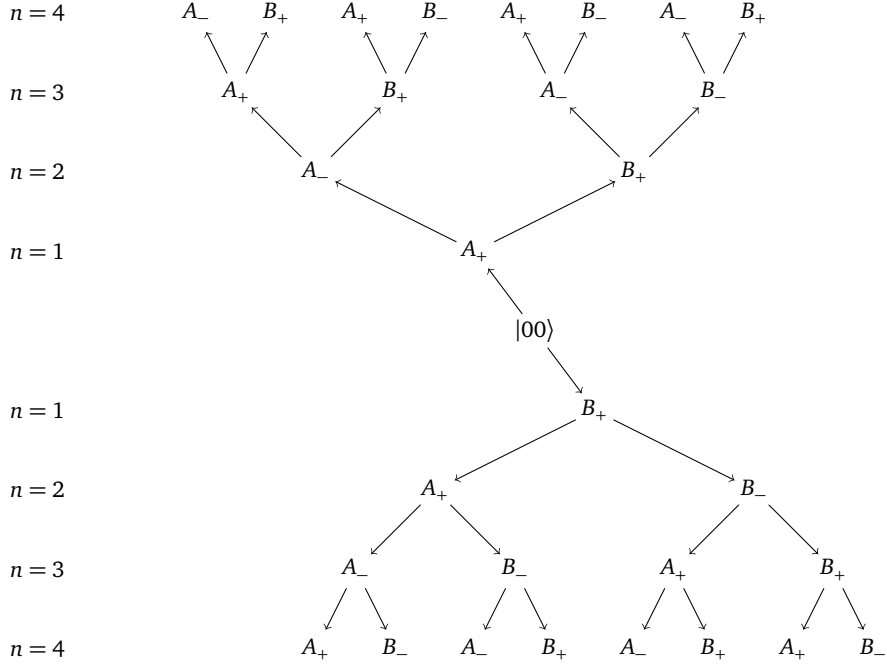}}
    \caption{A schematic depiction of the action of Eqs.~(\ref{eq:expEven}) and~(\ref{eq:expOdd}) for $n=4$ branchings. As in Fig.~\ref{fig:littlebintree} the arrows symbolize the action of the detector ladder operators $\s_{A,B}^\pm$. The left-pointing arrows correspond to the $A$ detector while the right pointing arrows are for the $B$ detector. The initial state is $\ket{00}$ and the symbols $A_\pm$ and $B_\pm$ are the boson operators `left behind' after the action of $\s_{A,B}^\pm$. So as we follow any path in the tree (up or down from $\ket{00}$) we collect the boson operators thus forming an operator sequence in Eqs.~(\ref{eq:expEven}) and~(\ref{eq:expOdd}). There are many paths with the same operator products and their multiplicities are the binomial coefficients derived in Proposition~\ref{prop:phaseI}.}
    \label{fig:bintree}
\end{figure}
%
  But this immediately provides the desired counting because we have just proved a bijection between the number of branches of a perfect binary tree of the depth $n$ and the number of non-zero strings of the same length. Consider~Eq.~(\ref{eq:allPaths2Atoms00}). We set the total number of branchings $n=2m$ and $p=\ell$ giving us $q=m-\ell$. The number of paths with $2\ell$ left branchings (one $\ell$ for $\s_A^+$ and one $\ell$ for $\s_A^-$ making it even since we have to end up in $\ket{0}_A$)  out of the total $n=2m$ branchings is $\binom{2m}{2\ell}$ and so is the number of nonzero strings in~Eq.~(\ref{eq:allPaths2Atoms00}). Identically, for~(\ref{eq:allPaths2Atoms11}) we again set $n=2m$ and  $p=\ell$ and obtain $\binom{2m}{2\ell-1}$ of nonzero strings ($2\ell-1$ left branchings and $m-\ell+1$ right branchings). For Eq.~(\ref{eq:allPaths2Atoms01}) we set $n=2m+1$ and $p=\ell$. Hence $\binom{2m+1}{2\ell}$ is the total number of nonzero strings. Finally, for Eq.~(\ref{eq:allPaths2Atoms10}) we have $n=2m+1$ and we now set $q=\ell$. In this way we obtain the same number of nonzero strings $\binom{2m+1}{2\ell}$ as in the previous case.

  The proposition statement is obtained by inspecting the LHS of Eqs.~(\ref{eq:expEven}) and~(\ref{eq:expOdd}) where the boson operators $A_\pm,B_\pm$ accompany the corresponding ladder operators and the RHS consist of the products of the boson operators from nonzero strings.
\end{proof}
\begin{rem}
    Fig.~\ref{fig:bintree} illustrates the proof. It also illustrates the fact that the same analysis can be done for any initial state $\ket{ij}$. Depending on the values of $i$ and $j$ the arrows in the binary tree will represent different ladder operators leading to the same counting and different boson operators on the RHS of  Eqs.~(\ref{eq:expEven}) and~(\ref{eq:expOdd}). In this way we are able to find all the unitary matrix elements and therefore we can act on any input state (pure or mixed) from a space spanned by $\ket{ij}_{AB}$.
\end{rem}
\begin{exa}
  Setting $m=0,1,2$ in Eq.~(\ref{eq:expEven}) and $m=0,1$ in Eq.~(\ref{eq:expOdd}) (corresponding to the perturbative order $n=4$) we quickly reproduce the fourth order calculation in~\cite{bradler2016absolutely} (Eqs.~(11)) and the result can be applied to a variety of situations~\cite{ver2009entangling,reznik2005violating,cliche2010information,lin2006accelerated}. With the same ease we can obtain an output state for $m\gg0$.

  A non-trivial check is the trace of~(\ref{eq:omega}) being equal to one irrespective of the perturbative order. Case $n=4$ was verified in~\cite{bradler2016absolutely} and a straightforward calculation shows that it is true for $n=6$ as well\footnote{Note that the fourth order is an absolute must if one wants to calculate any entropic quantity to properly assess the importance of entanglement for quantum communication. If only the second order is calculated (almost always the case with a notable exception of~\cite{bradler2016absolutely}) then the $\kbr{11}{11}$ component of~(\ref{eq:omega}) is zero and two of the eigenvalues are negative.}.
\end{exa}

\subsection*{Phase~II}
We found the expressions $ \bra{kl} U_AU_B \ket{00}$ and, crucially, the exponential growth in $n$ in Eq.~(\ref{eq:twoUDWdetectors}) was managed such that the number of summands in Eqs.~(\ref{eq:expEven}) and~(\ref{eq:expOdd}) increases polynomially with $n$. However, upon inserting two such expressions into Eq.~(\ref{eq:omega}) we note that the exponential growth crept back. This is because of the following result:
\begin{thm}[Isserlis'~\cite{isserlis1918formula}]\label{thm:Isserlis}
  Let $x_i$ be a Gaussian random variable satisfying $\bbE[\prod_{i=1}^{2m+1}x_i]=0$. Then
  \begin{equation}\label{eq:Isserlis}
    \bbE\Big[\prod_{i=1}^{2m}x_i\Big]=\sum_{r=1}^{(2m-1)!!}\prod_{\genfrac{}{}{0pt}{2}{j,k=1}{j<k}}^{m}\bbE_r[x_{j}x_{k}],
  \end{equation}
  where the sum goes over the products of bivariate expectation values $\bbE_r$.
\end{thm}
\begin{rem}[notational]
  In our case the Gaussian random variable will be a product of time-ordered free scalar fields $\phi_0(x_i)$ and the expectation value will be taken wrt to Minkowski (i.e non-interacting) vacuum $\ket{0_M}$:
  \begin{equation}\label{eq:notation}
     \bbE\Big[\prod_{i=1}^{2m}x_i\Big]=\bra{0_M}\mathsf{T}\big\{\prod_{i=1}^{2m}\phi_{0,i}\big\}\ket{0_M}\equiv\l\prod_{i=1}^{2m}i\ra_0,
  \end{equation}
  where  $\phi_{0,i}\equiv\phi_0(x_i)$ and on the RHS is the minimalist notation we will be mostly using.
\end{rem}
\begin{rem}
  The theorem is also known as Wick's theorem~\cite{wick1950evaluation}. Wick's theorem transforms time-ordered operator expression into a normal form~\cite{schwartz2014quantum} and Isserlis' result is recovered upon taking a (free) vacuum expectation value. As a matter of fact, Isserlis' theorem is more general since naturally there is no notion of time-ordering  in Eq.~(\ref{eq:Isserlis}) and so the theorem applies even for `ordinary' products of free scalar fields. Of course, we are actually calculating the time ordered version as it appears in the definition of the sought after Green's function.
\end{rem}
The exponential growth lies in the presence of the double factorial in~(\ref{eq:Isserlis}). For $2m$ different fields in Eq.~(\ref{eq:notation}) there is not much to do but the situation is different for a constant number of fields. We first prove the following result whose application goes beyond the problem solved in this paper (see Sec.~\ref{sec:phin}).
\begin{thm}\label{thm:phaseII}
  Let $\ell_i\geq0$ and $\sum_{i=1}^4\ell_i=2m$. Then Green's function $\l1^{\ell_1}2^{\ell_2}3^{\ell_3}4^{\ell_4}\ra_0$ can be written in terms of a polynomial number of products of two-point correlation functions:
  \begin{equation}\label{eq:GreensSplit}
    \l1^{\ell_1}2^{\ell_2}3^{\ell_3}4^{\ell_4}\ra_0
    \equiv\l\underbrace{1\dots1}_{\ell_1}\underbrace{2\dots2}_{\ell_2}\underbrace{3\dots3}_{\ell_3}\underbrace{4\dots4}_{\ell_4}\ra_0
    =\sum_{\ba}\mu_{\ba}\prod_{\genfrac{}{}{0pt}{2}{i,j=1}{i<j}}^{4}\l ij\ra^{\a_{ij}}_0,
  \end{equation}
  where $\mu_{\ba}$ is the product multiplicity, $\a_{ij}\in\bbZ_{\geq0}$ and  $\ba\df(\a_{11},\a_{12},\a_{13},\a_{14},\a_{22},\a_{23},\a_{24},\a_{33},\a_{34},\a_{44})$.
\end{thm}
\begin{defi}\label{def:coformations}
  Different products $\prod_{\genfrac{}{}{0pt}{2}{i,j=1}{i<j}}^{4}\l ij\ra^{\a_{ij}}_0$~will be called \emph{conformations} and $\ba$ a \emph{conformation exponent}.
\end{defi}
\begin{rem}
  At first sight there is nothing surprising about Eq.~(\ref{eq:GreensSplit}). Except for the double factorial the equation looks exactly like Eq.~(\ref{eq:Isserlis}) and it must be that way. After all, it is its special case. The novelty lies in the polynomial number of summands for a constant number of fields (four in this case) and in a constructive (and simple) way of obtaining all the coefficients $\a_{ij}$. The multiplicity factors $\mu_{\ba}$ are another crucial ingredient and they will be calculated in Theorem~\ref{thm:phaseIII}.
\end{rem}
\begin{proof}
  The solution of Eq.~(\ref{eq:GreensSplit}) can be found by solving the following system of four linear Diophantine equations for $\a_{ij}\in\bbZ_{\geq0}$:
\begin{subequations}\label{eq:Diophantine}
  \begin{align}
    2\a_{11}+\a_{12}+\a_{13}+\a_{14} & = \ell_1,\label{eq:Diophantine1} \\
    \a_{12}+2\a_{22}+\a_{23}+\a_{24} & = \ell_2,\label{eq:Diophantine2} \\
    \a_{13}+\a_{23}+2\a_{33}+\a_{34} & = \ell_3, \\
    \a_{14}+\a_{24}+\a_{34}+2\a_{44} & = \ell_4.\label{eq:Diophantine4}
  \end{align}
\end{subequations}
Linear Diophantine equations have zero or infinitely many solutions (counting both positive and negative ones). The numerical coefficients for the variables $\a_{ij}$ are so simple  that we can easily guess a solution for~(\ref{eq:Diophantine1}) (say $\a_{12}=\a_{13}=0$ and $\a_{14}=\ell_1-2\a_{11}$) and so there are infinitely many solutions. It also means that there are finitely many non-negative solutions we are after. How do we obtain them in a systematic way? Since $0\leq\a_{ii}\leq\lfloor\ell_i/2\rfloor$ and $0\leq\a_{ij}\leq\min{\{\ell_i,\ell_j\}}$ for $i\neq j$ we fix $\a_{11},\a_{12}$ (starting from $\a_{11}=\a_{12}=0$) and simply list all the admissible $\a_{13}$ and $\a_{14}$. We continue by increasing $\a_{12}$ by one and repeat the whole process and the same for $\a_{11}$. In the next step we move to Eqs.~(\ref{eq:Diophantine2})-(\ref{eq:Diophantine4}) by inserting all found triples $\{\a_{12},\a_{13},\a_{14}\}$ and repeating the procedure until we find all admissible solutions in the second row. We insert the found solutions to the third and fourth row until we find all solutions. The result is a complete list of ten-tuples $\boldsymbol{\a}$ solving Eqs.~(\ref{eq:Diophantine}).

How many different  $\boldsymbol{\a}$'s are there? We don't need an exact number (it is actually not that straightforward to count the solutions as we will see in Lemma~\ref{lem:estimateSols}) and a polynomial upper bound is enough\footnote{More precisely, we need to list all the solutions  $\ba$ and once we have them we simply count them. What we don't need for the purpose of this work is a closed expression giving the number of non-negative solutions of~(\ref{eq:Diophantine}).}. The number of solutions for Eq.~(\ref{eq:allPaths2Atoms00}) depends on the parity of $\ell_1$. For $\ell_1$ even and $\a_{11}=0$ the number of nonnegative solutions is a triangle number $(2\ell_1+1)(2\ell_1+2)/2$ and by incrementing $\a_{11}$ by one up to $\ell_1/2$ we get the rest of solutions. Summing them up we obtain a total of
\begin{equation}\label{eq:1lineSols}
  \sum_{s=0}^{\ell_1/2}{(2s+1)(2s+2)\over2}={1\over24}(2+\ell_1)(4+\ell_1)(3+2\ell_1)
\end{equation}
solutions. A similar analysis leads to the following number of solutions for an odd $\ell_1$:
\begin{equation}
  \sum_{s=0}^{(\ell_1+1)/2}{2s(2s+1)\over2}={1\over24}(1+\ell_1)(3+\ell_1)(7+2\ell_1).
\end{equation}
If all equations~(\ref{eq:Diophantine}) were independent we would get a scaling for the number of solutions like $\Ocal(\ell^{12})$, where $\ell=\max_i{\ell_i}$. This is enough for the proof but it grossly overestimates the number of solutions. A more careful counting gives us the $\Ocal(\ell^{6})$ scaling.
\end{proof}
\begin{rem}
  There is an extensive literature on how to solve a system of linear Diophantine equations~\cite{schrijver1998theory,antsaklis2006linear}. But the question is not really whether one can solve it but how fast (see, for example,~\cite{kannan1979polynomial}). It is likely that the simple form of our system does not even require any sophisticated method to solve it. One just systematically lists all the solutions as sketched in the proof.
\end{rem}
Even a better estimate is given in the next lemma.
\begin{lem}\label{lem:estimateSols}
  The number of non-negative solutions to~(\ref{eq:Diophantine}) is upper bounded by
  \begin{equation}\label{eq:finalEstimate}
    {1\over2^5}(\ell+2)^5(\ell+1),
  \end{equation}
  where $\ell=\max{\ell_i}$.
\end{lem}
\begin{proof}
  To simplify the derivation, assume $\ell_i=\max_i{\ell_i}\equiv\ell$ and if necessary let's promote $\ell$ to be the closest (greater) even number. From all admissible $\a_{ii}$ it is $\a_{ii}=0,\forall i,$  yielding the greatest number of solutions so instead of~(\ref{eq:Diophantine}) let's start by studying the following simpler Diophantine system:
    \begin{subequations}\label{eq:DiophantineSim}
      \begin{align}
        \a_{12}+\a_{13}+\a_{14} & = \ell,\label{eq:Diophantine1Sim} \\
        \a_{12}+\a_{23}+\a_{24} & = \ell,\label{eq:Diophantine2Sim} \\
        \a_{13}+\a_{23}+\a_{34} & = \ell,\label{eq:Diophantine3Sim} \\
        \a_{14}+\a_{24}+\a_{34} & = \ell.\label{eq:Diophantine4Sim}
      \end{align}
    \end{subequations}
  Eqs.~(\ref{eq:DiophantineSim}) inherited an interesting structure from~(\ref{eq:Diophantine}). Starting with~(\ref{eq:Diophantine2Sim}), every row shares exactly one variable (but always different) from all previous rows. From the last two lines of~(\ref{eq:DiophantineSim}) we observe  $\a_{13}+\a_{23}=\a_{14}+\a_{24}$. Using the  first two lines of~(\ref{eq:DiophantineSim}) we also find $\a_{13}+\a_{14}=\a_{23}+\a_{24}$ and so we get
  \begin{subequations}\label{eq:firstlast}
    \begin{align}
      \a_{13} & =\a_{24}, \\
      \a_{14} & =\a_{23}
    \end{align}
  \end{subequations}
  as the only consistent solutions. Then for any allowed pair of coefficients $\{\a_{13},\a_{14}\}$ (allowed means $\a_{13}+\a_{14}=L$ for $0\leq L\leq\ell$ and $L\in\bbZ_{\geq0}$) there is one pair $\{\a_{23},\a_{24}\}$ that can satisfy~(\ref{eq:firstlast}). This is because even though there are two other free variables, $\a_{12}$ and $\a_{34}$, the permutation symmetry of~(\ref{eq:DiophantineSim}) together with Eqs.~(\ref{eq:firstlast}) imply
  $$
  \a_{12}=\a_{34}.
  $$
  Finally, looking at~(\ref{eq:Diophantine1Sim}), we find $\mu[\a_{12}=L]=\ell+1-L$, where $\mu$ stands for the multiplicity of $\a_{12}$ with the value $L$, and so there is
  \begin{equation}\label{eq:finalCountinga00}
    \sum_{L=0}^\ell(\ell+1-L)={1\over2}(\ell+1)(\ell+2)
  \end{equation}
  solutions of~(\ref{eq:DiophantineSim}). By upper-bounding all $(\ell/2+1)^4$  systems of linear Diophantine equations by~\eqref{eq:finalCountinga00} we get~\eqref{eq:finalEstimate}. This scales as $\Ocal((\ell+2)^6)$.
\end{proof}
\begin{exa}\label{exa:ell2222}
  The number of conformations of $\l1^{2}2^{2}3^{2}4^{2}\ra_0$ is 17 and their exponents read
    \begin{equation}\label{eq:ell2222}
    \{\ba_c\}_{c=1}^{17}=\left\{\begin{array}{cccccccccc}
     (0 ,& 0 ,& 0 ,& 2 ,& 1 ,& 0 ,& 0 ,& 1 ,& 0 ,& 0) \\
     (1 ,& 0 ,& 0 ,& 0 ,& 0 ,& 0 ,& 2 ,& 1 ,& 0 ,& 0) \\
     (1 ,& 0 ,& 0 ,& 0 ,& 1 ,& 0 ,& 0 ,& 0 ,& 2 ,& 0) \\
     (1 ,& 0 ,& 0 ,& 0 ,& 1 ,& 0 ,& 0 ,& 1 ,& 0 ,& 1) \\
     (0 ,& 1 ,& 0 ,& 1 ,& 0 ,& 0 ,& 1 ,& 1 ,& 0 ,& 0) \\
     (0 ,& 2 ,& 0 ,& 0 ,& 0 ,& 0 ,& 0 ,& 0 ,& 2 ,& 0) \\
     (0 ,& 2 ,& 0 ,& 0 ,& 0 ,& 0 ,& 0 ,& 1 ,& 0 ,& 1) \\
     (1 ,& 0 ,& 0 ,& 0 ,& 0 ,& 1 ,& 1 ,& 0 ,& 1 ,& 0) \\
     (0 ,& 1 ,& 0 ,& 1 ,& 0 ,& 1 ,& 0 ,& 0 ,& 1 ,& 0) \\
     (0 ,& 0 ,& 0 ,& 2 ,& 0 ,& 2 ,& 0 ,& 0 ,& 0 ,& 0) \\
     (1 ,& 0 ,& 0 ,& 0 ,& 0 ,& 2 ,& 0 ,& 0 ,& 0 ,& 1) \\
     (0 ,& 0 ,& 1 ,& 1 ,& 1 ,& 0 ,& 0 ,& 0 ,& 1 ,& 0) \\
     (0 ,& 1 ,& 1 ,& 0 ,& 0 ,& 0 ,& 1 ,& 0 ,& 1 ,& 0) \\
     (0 ,& 0 ,& 1 ,& 1 ,& 0 ,& 1 ,& 1 ,& 0 ,& 0 ,& 0) \\
     (0 ,& 1 ,& 1 ,& 0 ,& 0 ,& 1 ,& 0 ,& 0 ,& 0 ,& 1) \\
     (0 ,& 0 ,& 2 ,& 0 ,& 0 ,& 0 ,& 2 ,& 0 ,& 0 ,& 0) \\
     (0 ,& 0 ,& 2 ,& 0 ,& 1 ,& 0 ,& 0 ,& 0 ,& 0 ,& 1)
    \end{array}
    \right\}
    \end{equation}
\end{exa}
\begin{rem}
  A calculation performed in~\cite{bradler2016dio} leads to a closed expression for the number of non-negative solutions of~\eqref{eq:Diophantine} for $\ell_i=\ell$ even to be
  \begin{equation}\label{eq:ellEven}
    \mathsf{e}=\frac{1}{576} (\ell+2) (\ell+4)\big(\ell (\ell+5) (\ell (\ell+4)+12)+72\big).
  \end{equation}
  For $\ell=2$ we get $\mathsf{e}=17$ from the previous example.
\end{rem}
\subsection*{Phase~III} Let's find the expression for the conformation multiplicity factor $\mu_{\ba}$. First an auxiliary lemma.
\begin{lem}\label{lem:fcnCardinality}
  Let $S$ and $T$ be discrete sets where $|S|=s,|T|=t$. Then there is
  \begin{equation}\label{eq:fcnCardinality}
    \binom{s}{n}\,t\times\ldots\times(t-n+1)
  \end{equation}
  bijective functions $f:X\mapsto Y$ where $X\subset S,Y\subset T$ and $|X|=|Y|=n$ where $0<n\leq\min{\{s,t\}}$. For $n=0$ we set Eq.~(\ref{eq:fcnCardinality}) to one.
\end{lem}
\begin{rem}
  If $s=t=n$ then~(\ref{eq:fcnCardinality}) becomes $n!$ which is  known to be the number of bijections from a set to itself (the number of permutations). To make the notation more concise in following text we will use the definition of the falling factorial: $(m)_n\df m\times\ldots\times (m-n+1)$. So Eq.~(\ref{eq:fcnCardinality}) is written as $\binom{s}{n}(t)_n$.
\end{rem}
\begin{proof}
  The coefficient $\binom{s}{n}$ is simply a number of all possible domains $X\subset S$ whose cardinality is $n$. Then for every domain $X$ there is $\binom{t}{n}n!$ codomains $Y\subset T$ of the same cardinality. The  coefficient $n!$ comes from the number of permutations within each of $\binom{t}{n}$ codomains $Y$.
\end{proof}
\begin{thm}\label{thm:phaseIII}
  For a given conformation exponent $\boldsymbol{\a}$ the multiplicity factor reads
  \begin{equation}\label{eq:multFactor}
    \mu_{\ba}=\G\g_{12}\g_{13}\g_{14}\g_{23}\g_{24}\g_{34},
  \end{equation}
  where
  \begin{subequations}\label{eq:Gammas}
    \begin{align}
      \G & = \binom{\ell_1}{2\a_{11}}\binom{\ell_2}{2\a_{22}}\binom{\ell_3}{2\a_{33}}\binom{\ell_4}{2\a_{44}}\prod_{i=1}^4(2\a_{ii}-1)!!,\label{eq:GammasGAMMA}\\
      \g_{12} & = \binom{\ell_1-2\a_{11}}{\a_{12}}(\ell_2-2\a_{22})_{\a_{12}},\label{eq:Gammasga12}\\
      \g_{13} & = \binom{\ell_1-2\a_{11}-\a_{12}}{\a_{13}}(\ell_3-2\a_{33})_{\a_{13}},\label{eq:Gammasga13}\\
      \g_{14} & = \binom{\ell_1-2\a_{11}-\a_{12}-\a_{13}}{\a_{14}}(\ell_4-2\a_{44})_{\a_{14}},\\
      \g_{23} & = \binom{\ell_2-2\a_{22}-\a_{12}}{\a_{23}}(\ell_3-2\a_{33}-\a_{13})_{\a_{23}}, \\
      \g_{24} & = \binom{\ell_2-2\a_{22}-\a_{12}-\a_{23}}{\a_{24}}(\ell_3-2\a_{44}-\a_{14})_{\a_{24}},\\
      \g_{34} & = \binom{\ell_3-2\a_{33}-\a_{13}-\a_{23}}{\a_{34}}(\ell_4-\a_{44}-\a_{14}-\a_{24})_{\a_{34}}.
    \end{align}
  \end{subequations}
\end{thm}
\begin{proof}
  We will use repeatedly Lemma~\ref{lem:fcnCardinality} by identifying $\l ij\ra_0^{\a_{ij}}$ for $i\neq j$ from Theorem~\ref{thm:phaseII} in the following way:
  \begin{align}
    n & =\a_{ij}, \\
    s & =\ell_i,\\
    t & =\ell_j.
  \end{align}
  So $S$ is the set of all $i$'s and $T$ is the set of all $j$'s. However, if we do more than one two-point Green's functions we have to take into account that the sets $S$ and $T$ might have shrunk. This depends on whether in the preceding Green's function $\l kl\ra_0^{\a_{kl}}$ we had $k=i$ or $l=j$. The strategy we will follow here is to first count the multiplicity of $\l ii\ra_0^{\a_{ii}}$ as they are independent (meaning non-overlapping for different $i$'s) and then the multiplicities of $\l 12\ra_0^{\a_{12}},\l 13\ra_0^{\a_{13}},\l 14\ra_0^{\a_{14}},\l 23\ra_0^{\a_{23}},\l 24\ra_0^{\a_{24}}$ and $\l 34\ra_0^{\a_{34}}$ in this order(!).

  The multiplicity of  $\l ii\ra_0^{\a_{ii}}$ immediately follows from Isserlis' theorem. Looking at Eq.~(\ref{eq:Isserlis}) we see that if $x_i=x_j,\forall j$ there will be $(2m-1)!!$ identical products. Hence ($m=\a_{ii}$)
  \begin{equation}\label{eq:selfMultiplicity}
    \g_{ii}=\binom{\ell_i}{2\a_{ii}}(2\a_{ii}-1)!!,
  \end{equation}
  as follows from Eq.~(\ref{eq:fcnCardinality})\footnote{Indeed, the double factorial follows either from Theorem~\ref{thm:Isserlis} as stated or the counting in Lemma~\ref{lem:fcnCardinality}.}. The factor of two in the binomial `denominator' accounts for the two $i$'s in  $\l ii\ra_0$. Repeating this procedure for all $i$ we get $\G=\g_{11}\g_{22}\g_{33}\g_{44}$ in Eq.~(\ref{eq:GammasGAMMA}). To get $\g_{ij}$ we then have to keep track of the set cardinality. To get $\g_{12}$ (using the notation of Lemma~\ref{lem:fcnCardinality}) we set $n=\a_{12}$ and notice that the cardinality of the discrete set of `ones' is decreased by those contributing to $\g_{11}$. So $s=\ell_1-2\a_{11}$. Similarly, the cardinality of the `twos' is decreased by the number of elements contributing to $\g_{22}$. So $t=\ell_2-2\a_{22}$ and $\g_{12}$ in (\ref{eq:Gammasga12}) follows. At this point it is evident how to proceed so let's only derive $\g_{13}$ which is next in line. We set $n=\a_{13}$ and the set of `ones' is now smaller by those contributing to $\g_{11}$ but also to those from $\g_{12}$ calculated in the previous step. Hence $s=\ell_1-2\a_{11}-\a_{12}$. The set of `threes' was decreased only by the elements contributing to $\g_{33}$ and so $t=\ell_3-2\a_{33}$ and Eq.~(\ref{eq:Gammasga13}) follows. Going up to $\g_{34}$ we arrive at the overall multiplicity factor, Eq.~(\ref{eq:multFactor}).
\end{proof}
\begin{rem}
    We emphasize the importance of the order in which we follow the construction of $\g_{ij}$. Other possibilities are certainly plausible but we believe the sketched one is the most intuitive.
\end{rem}
\begin{exa}(Important)
  Let's find the multiplicities from Example~\ref{exa:ell2222}. By setting $\ell_i=2$ and plugging $\{\boldsymbol{\a}_c\}_{c=1}^{17}$ into Eqs.~(\ref{eq:Gammas}) we get
    \begin{equation}\label{eq:multiell2222}
      \mu_{\boldsymbol{\a}_c}=(2, 2, 2, 1, 8, 4, 2, 8, 16, 4, 2, 8, 16, 16, 8, 4, 2),
    \end{equation}
    where $c=1,\dots,17$. Crucially (and this is a highly non-trivial check of both Theorems~\ref{thm:phaseII} and~\ref{thm:phaseIII}!), we get
    \begin{equation}\label{eq:checkell2222}
      \sum_{c=1}^{17}\mu_{\boldsymbol{\a}_c}=105=7!!=\Big(\sum_{i=1}^{4}\ell_i-1\Big)!!.
    \end{equation}
    According to Eq.~(\ref{eq:Isserlis}) we have $\sum_{i=1}^4\ell_i=2m$ and so the RHS of~(\ref{eq:checkell2222}) is the total number of products as it should be according to Isserlis' theorem.
\end{exa}
\subsection*{Putting all pieces together}
We use Proposition~\ref{prop:phaseI} to insert Eqs.~(\ref{eq:expEven}) and~(\ref{eq:expOdd}) to Eq.~(\ref{eq:omega}) and get the output state components to the $n$-th order. Let's take a look at the $\kbr{00}{00}$ component:
\begin{align}\label{eq:omega00}
  \om_{AB}(\kbr{00}{00})
  &= \bra{0_M}\mathsf{T}\Big\{
    \sum_{m+m'=0}^{n}{\la^{2(m+m')}(-1)^{m+m'}\over (2m)!(2m')!}\bigg(\sum_{\ell=0}^m\binom{2m}{2\ell}\,A_+^\ell A_-^\ell B_+^{m-\ell}B_-^{m-\ell}\bigg)\nn\\
    &\quad\times
    \bigg(\sum_{\ell'=0}^{m'}\binom{2m'}{2\ell'}\,A_+^{\ell'} A_-^{\ell'} B_+^{m'-\ell'}B_-^{m'-\ell'}\bigg)^\dg
  \Big\} \ket{0_M}\nn\\
  &= \bra{0_M}\mathsf{T}{\Big\{
    \sum_{m+m'=0}^{n}{\la^{2(m+m')}(-1)^{m+m'}\over (2m)!(2m')!}
    \binom{2m}{2\ell}\binom{2m'}{2\ell'}
    \,A_+^{\ell+\ell'} A_-^{\ell+\ell'} B_+^{m+m'-\ell-\ell'}B_-^{m+m'-\ell-\ell'}
  \Big\}} \ket{0_M}\nn\\
    &= \sum_{m+m'=0}^{n}{\la^{2(m+m')}(-1)^{m+m'}\over (2m)!(2m')!}
    \binom{2m}{2\ell}\binom{2m'}{2\ell'}
    \l1^{\ell+\ell'}2^{\ell+\ell'}3^{m+m'-\ell-\ell'}4^{m+m'-\ell-\ell'}\ra_0.
\end{align}
We identified  $A_+\rightleftharpoons1,A_-\rightleftharpoons2,B_+\rightleftharpoons3$ and $B_-\rightleftharpoons4$ and in this form it is ready for Theorem~\ref{thm:phaseII} and onto Theorem~\ref{thm:phaseIII}. The number of terms in~Eq.~(\ref{eq:omega00}) is polynomial in $m$ and therefore in $n$ as well. By setting $\ell_1=\ell_2=\ell+\ell',\ell_3=\ell_4=m+m'-\ell-\ell'$ and taking the greatest $\ell_i$, Lemma~\ref{lem:estimateSols} gives us a tighter polynomial upper bound on the number of solutions compared to Theorem~\ref{thm:phaseII}. We thus arrived at the main result of the paper.

\section{Interacting fields in the $\phi^n$ scalar model to the second perturbative order}\label{sec:phin}

One of the simplest interacting theory is given by the Lagrangian
\begin{equation}\label{eq:Lagrphi34}
  \Lcal=\Lcal_0+\Lcal_{\mathrm{int}}={1\over2}\partial^\nu\phi\partial_\nu\phi-{1\over2}m^2\phi^2-g{\phi^n\over n!},
\end{equation}
where $\phi$ is a real  scalar field and $g$ a coupling constant. The model is not exactly solvable and  its perturbative evolution, Feynman rules and much more (such as renormalization) are typically trained on $n=3$~\cite{schwartz2014quantum} or on a more physical model $n=4$. The quantity of interest is again the $n$-point Green's  function $\bra{\Om}\mathsf{T}\{\phi(x_1)\dots\phi(x_n)\}\ket{\Om}$ for interacting fields and its expectation value is calculated wrt to the interacting vacuum $\ket{\Om}$. Isserlis's or Wick's theorem cannot be applied for interacting fields but from the Lagrangian (or Feynman path integral) formulation~\cite{schwartz2014quantum} we know that the interacting Green's function satisfies
\begin{equation}\label{eq:GreenInteracting}
  \bra{\Om}\mathsf{T}\{\phi(x_1)\dots\phi(x_k)\}\ket{\Om}={\bra{0_M}\mathsf{T}\{\phi_0(x_1)\dots\phi_0(x_k)\exp{[{i\over n!}\int \diff{4}{x}\phi_0^n]\}}\ket{0_M}\over\bra{0_M}\mathsf{T}\{\exp{[{i\over n!}\int\diff{4}{x}\phi_0^n]}\}\ket{0_M}}.
\end{equation}
Interestingly, for $k=2$ Theorems~\ref{thm:phaseII} and~\ref{thm:phaseIII} can be used to an automatic calculation up to the second order for any $n$ as we will now illustrate for $n=3$ and $n=4$.
\begin{exa}[$\phi^3$ for $k=2$]
  We expand the numerator of (\ref{eq:GreenInteracting}) as is customary~\cite{schwartz2014quantum,peskin1995introduction} and take the third term. It is proportional to
  \begin{equation}\label{eq:phi3expansionNUM}
    \int\diff{4}{x}\int\diff{4}{y}\bra{0_M}\mathsf{T}\{\phi_0(x_1)\phi_0(x_2)\phi_0(x)^3\phi_0(y)^3\}\ket{0_M}
    =  \int\diff{4}{x}\int\diff{4}{y} \l1^{1}2^{1}3^{3}4^{3}\ra_0,
  \end{equation}
  where on the RHS we use the notation of Theorem~\ref{thm:phaseII} by setting $x_1\leftrightharpoons1,x_2\leftrightharpoons2,x\leftrightharpoons3$ and $y\leftrightharpoons4$. Similarly the second term of the denominator is proportional to
  \begin{equation}\label{eq:phi3expansionDENOM}
    \int\diff{4}{x}\int\diff{4}{y}\bra{0_M}\mathsf{T}\{\phi_0(x)^3\phi_0(y)^3\}\ket{0_M}
    =  \int\diff{4}{x}\int\diff{4}{y} \l1^{0}2^{0}3^{3}4^{3}\ra_0.
  \end{equation}
  Using Theorem~\ref{thm:phaseII} we find 8 conformations for the correlator in Eq.~(\ref{eq:phi3expansionNUM}) whose exponents read
  \begin{equation}\label{eq:ell1133}
    \{\ba_c\}_{c=1}^{8}=\left\{\begin{array}{cccccccccc}
     (0 ,& 0 ,& 0 ,& 1 ,& 0 ,& 0 ,& 1 ,& 1 ,& 1 ,& 0) \\
     (0 ,& 1 ,& 0 ,& 0 ,& 0 ,& 0 ,& 0 ,& 0 ,& 3 ,& 0) \\
     (0 ,& 1 ,& 0 ,& 0 ,& 0 ,& 0 ,& 0 ,& 1 ,& 1 ,& 1)\\
     (0 ,& 0 ,& 0 ,& 1 ,& 0 ,& 1 ,& 0 ,& 0 ,& 2 ,& 0) \\
     (0 ,& 0 ,& 0 ,& 1 ,& 0 ,& 1 ,& 0 ,& 1 ,& 0 ,& 1) \\
     (0 ,& 0 ,& 1 ,& 0 ,& 0 ,& 0 ,& 1 ,& 0 ,& 2 ,& 0) \\
     (0 ,& 0 ,& 1 ,& 0 ,& 0 ,& 0 ,& 1 ,& 1 ,& 0 ,& 1) \\
     (0 ,& 0 ,& 1 ,& 0 ,& 0 ,& 1 ,& 0 ,& 0 ,& 1 ,& 1)
    \end{array}\right\}
  \end{equation}
  and Theorem~\ref{thm:phaseIII} yields their multiplicities
  \begin{equation}\label{eq:multiell1133}
      \mu_{\ba_c}=(18, 6, 9, 18, 9, 18, 9, 18).
  \end{equation}
  We check the solution by inspecting
  \begin{equation}\label{eq:checkell1133}
      \sum_{c=1}^{8}\mu_{\ba_c}=105=7!!=\Big(\sum_{i=1}^{4}\ell_i-1\Big)!!,
  \end{equation}
  where the RHS follows from $\ell_1=\ell_2=1$ and $\ell_3=\ell_4=3$.

  For the denominator we get
  \begin{equation}
    \{\ba_c\}_{c=1}^{2}=\left\{
    \begin{array}{cccccccccc}
     (0 ,& 0 ,& 0 ,& 0 ,& 0 ,& 0 ,& 0 ,& 0 ,& 3 ,& 0) \\
     (0 ,& 0 ,& 0 ,& 0 ,& 0 ,& 0 ,& 0 ,& 1 ,& 1 ,& 1)
    \end{array}
    \right\}
  \end{equation}
  and
  \begin{equation}\label{eq:multiell0033}
      \mu_{\ba_c}=(6,9).
  \end{equation}
  We again check the consistence
  \begin{equation}\label{eq:checkell1133}
      \sum_{c=1}^{2}\mu_{\ba_c}=\Big(\sum_{i=1}^{4}\ell_i-1\Big)!!=15.
  \end{equation}
  By comparing with~\cite{schwartz2014quantum}(ch. 7) we also find a perfect agreement. From this point we proceed as usual. Since we integrate over the `internal' variables (here labelled 3 and 4), some conformations represent the same physical process. We can easily find them from (\ref{eq:ell1133}) as those invariant wrt the swap $3\leftrightharpoons4$. Those are the following pairs of ten-tuples: $(\ba_1,\ba_8)$, $(\ba_4,\ba_6)$ and $(\ba_5,\ba_7)$. Considering the double multiplicities we reproduce and finalize the derivation of the second order of the $\phi^3$ theory~\cite{schwartz2014quantum}(ch. 7).
\end{exa}
Of course, these are standard results obtainable through Wick's theorem. Yet, the ease with which we reproduced them is notable  as we don't need to count the contractions or draw the Feynman diagrams. An implementation of  Theorems~\ref{thm:phaseII} and~\ref{thm:phaseIII} in a suitable programming language such as Mathematica is straightforward. For example, one can convert the proof of Theorem~\ref{thm:phaseII} into a program that systematically finds all solutions. Essentially, if it takes one time step to find the first solution then all $t$ solutions can be find in $t$ time steps where $t$ increases polynomially. There is really no Diophantine calculation involved since due to the simple form of~(\ref{eq:Diophantine}) we get \emph{all} the solutions by inspection and only need to list them (save them to a memory).

The whole process would arguably be  much  more impressive for more than two interacting fields and for high-loop corrections. In that case the number of variables increases and Theorem~\ref{thm:phaseII} needs to be generalized. This can be done in a straightforward way but the presence of more variables is ultimately responsible for the factorial growth of the perturbative contributions. We obviously cannot get rid of the factorial growth but the easy-to-see character of the Diophantine linear system remains and so all the solutions can again be listed in $t$ time steps ($t$ grows factorially with a perturbative order, though). In this way, we can again get the perturbative contributions for any number of interacting fields, to a essentially an arbitrarily high perturbative order (keeping in mind the exponential growth discussed in this paragraph) and for all $n$ with  no effort whatsoever.

Furthermore, the Feynman diagrams are encoded in the conformation exponents. This is not surprising as $\l ij\ra_0$ represents a free propagator $\Delta(x_i-x_j)$ for $i,j$ being the internal degrees of freedom. So we easily recover them  from our approach if they are still needed. Let's illustrate it on our next example.
\begin{exa}[$\phi^4$ for $k=2$]
  Let's calculate the third term from the expanded numerator of~Eq.~(\ref{eq:GreenInteracting}). It is proportional to
  \begin{equation}\label{eq:phi4expansionNUM}
    \int\diff{4}{x}\int\diff{4}{y}\bra{0_M}\mathsf{T}\{\phi_0(x_1)\phi_0(x_2)\phi_0(x)^4\phi_0(y)^4\}\ket{0_M}
    =  \int\diff{4}{x}\int\diff{4}{y} \l1^{1}2^{1}3^{4}4^{4}\ra_0.
  \end{equation}
  As before, Theorem~\ref{thm:phaseII} reveals all the conformation exponents
  \begin{equation}\label{eq:ell1144}
   \{\ba_c\}_{c=1}^{11}=\left\{
    \begin{array}{cccccccccc}
     (0 ,& 0 ,& 0 ,& 1 ,& 0 ,& 0 ,& 1 ,& 1 ,& 2 ,& 0) \\
     (0 ,& 0 ,& 0 ,& 1 ,& 0 ,& 0 ,& 1 ,& 2 ,& 0 ,& 1) \\
     (0 ,& 1 ,& 0 ,& 0 ,& 0 ,& 0 ,& 0 ,& 0 ,& 4 ,& 0) \\
     (0 ,& 1 ,& 0 ,& 0 ,& 0 ,& 0 ,& 0 ,& 1 ,& 2 ,& 1) \\
     (0 ,& 1 ,& 0 ,& 0 ,& 0 ,& 0 ,& 0 ,& 2 ,& 0 ,& 2) \\
     (0 ,& 0 ,& 0 ,& 1 ,& 0 ,& 1 ,& 0 ,& 0 ,& 3 ,& 0) \\
     (0 ,& 0 ,& 0 ,& 1 ,& 0 ,& 1 ,& 0 ,& 1 ,& 1 ,& 1) \\
     (0 ,& 0 ,& 1 ,& 0 ,& 0 ,& 0 ,& 1 ,& 0 ,& 3 ,& 0) \\
     (0 ,& 0 ,& 1 ,& 0 ,& 0 ,& 0 ,& 1 ,& 1 ,& 1 ,& 1) \\
     (0 ,& 0 ,& 1 ,& 0 ,& 0 ,& 1 ,& 0 ,& 0 ,& 2 ,& 1) \\
     (0 ,& 0 ,& 1 ,& 0 ,& 0 ,& 1 ,& 0 ,& 1 ,& 0 ,& 2) \\
    \end{array}
    \right\}
  \end{equation}
  and Theorem~\ref{thm:phaseIII} their multiplicities
  \begin{equation}\label{eq:multiell1144}
    \mu_{\ba_c}=(144, 36, 24, 72, 9, 96, 144, 96, 144, 144, 36)
  \end{equation}
  followed by the obligatory check $ \sum_{c=1}^{11}\mu_{\ba_c}=\Big(\sum_{i=1}^{4}\ell_i-1\Big)!!=9!!$.   The unique conformations are depicted in Fig.~\ref{fig:feynmanphi4}. Hence we can immediately obtain the standard way of performing the perturbative expansion of the $\phi^4$ model (for two interacting fields and to the second order). The same is true for any $\phi^n$.
\end{exa}
    \begin{figure}[t]
   \fcolorbox{white}{white}{\includegraphics[scale=.8]{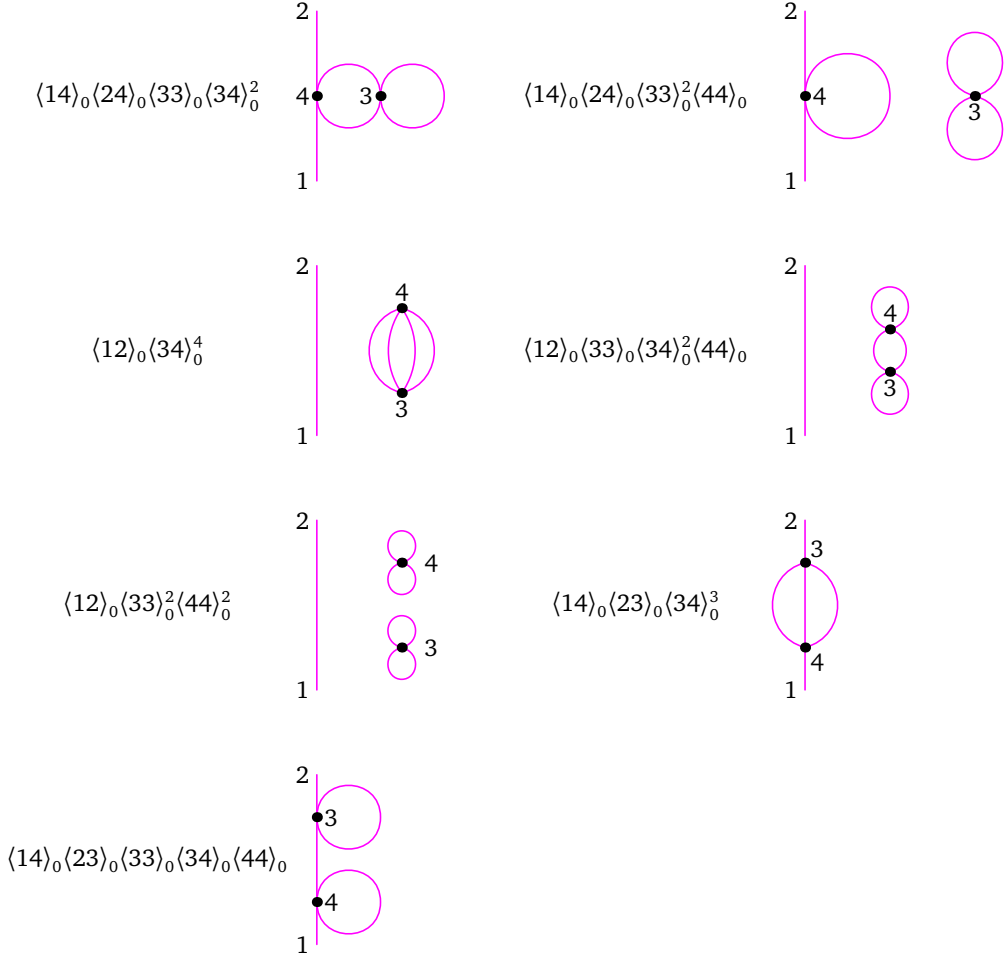}}
    \caption{Seven unique conformations (for the used convention see Definition~\ref{def:coformations} on page~\pageref{def:coformations}) and their Feynman diagrams out of eleven conformations~Eq.~(\ref{eq:phi4expansionNUM}) splits into. Conformation exponents~(\ref{eq:ell1144}) provide a succinct way of encoding Feynman's drawings. The dots (labeled $3$ and $4$) are the internal coordinates $x$ and $y$ we integrate over in~(\ref{eq:phi4expansionNUM}). The remaining four conformations are obtained by swapping 3 and 4 and therefore represent the same physical process.}
    \label{fig:feynmanphi4}
   \end{figure}
%

\section*{Conclusions and open problems}

In this paper we present two results: the main one is an efficient perturbative calculation of the interaction Hamiltonian for a pair of Unruh-DeWitt detectors in Minkowski spacetime. The result is independent on the detector details and works for any smearing, envelope function or trajectory by efficiently  decomposing a generic $2n$-point correlator to a product of two-point Green's functions where the detector details play a role. An efficient calculation means that the number of perturbative terms increases polynomially with the perturbative order of the coupling constant. This is atypical in quantum field theory as even in the simplest models (such as the interacting scalar $\phi^n$ model) the number of perturbative terms explodes factorially.

A notable accompanying result is an extremely straightforward calculation of the perturbative expansion for the previously mentioned interacting scalar $\phi^n$ model without going through the construction of Feynman diagrams (they can be immediately recovered from our approach). The presented method seems novel and even though the usual route through the creation of Feynman diagrams can be computerized  the proposed technique here mechanized the whole perturbative calculation to the point of no effort. We demonstrate it on the typical case of $n=3,4$ for the  $\phi^n$ model, for two interacting fields and up to the second order. A generalization for many interacting fields and a high perturbative order (again for any $n$) seems straightforward. This does not say that we can avoid the factorial growth of perturbative contributions -- we can just streamline it as much as possible. The details of this process can be an interesting direction to pursue.

Going back to the main result, it also offers several routes for further explorations. One is the study of more than two Unruh-DeWitt detectors and showing the same tractability as presented here for two detectors. It seems that for a fixed number of Unruh-DeWitt detectors the growth remains polynomial. Also, the technique presented here could be adopted for a detection of spinor fields obeying the fermion statistics~\cite{louko2016unruh}.

\bibliographystyle{unsrt}


\end{document}